\documentclass[a4paper,twocolumn,11pt,unpublished]{quantumarticle}
\pdfoutput=1

\usepackage{mathtools}
\usepackage{amsmath}
\usepackage{amssymb}
\usepackage{amsthm}
\usepackage[normalem]{ulem}
\usepackage{comment}
\usepackage{env}
\usepackage{enumitem}

\newtheorem{protocol}{Protocol}
\newcommand{\Accept}{\mathsf{Accept}}
\newcommand{\Reject}{\mathsf{Reject}}
\newcommand{\Pauli}{\mathcal P}
\newcommand{\Hilbert}{\mathcal H}
\newcommand{\Clifford}{\mathcal C}
\newcommand{\sgn}{\mathrm{sgn}}
\newcommand{\np}{m}
\newcommand{\Niid}{N_{\mathrm{iid}}}
\newcommand{\Nnoniid}{N_{\mathrm{non\text{-}iid}}}
\newcommand{\CEPS}{$\mathcal C$PS}
\newcommand{\mCEPS}{\mathcal{C}\text{PS}}
\newcommand{\STAB}{STAB}
\newcommand{\GS}{GS}
\newcommand{\HGS}{HGS}
\newcommand{\pdist}{\mathcal D} %
\newcommand{\jdist}{\mathcal D} %
\newcommand{\dist}{\mathcal G} %
\newcommand{\copyindex}{j}
\newcommand{\qubit}{i}
\renewcommand{\O}{\mathcal O}
\newcommand{\ketbra}[1]{\ket{#1}\!\bra{#1}}
\newcommand{\abs}[1]{\left|#1\right|}
\title{Efficient certification of intractable quantum states with few Pauli measurements}

\author{Sami Abdul Sater}
\author{Maxime Garnier}
\author{Thierry Martinez}
\author{Harold Ollivier}
\author{Ulysse Chabaud}
\affiliation{INRIA, DIENS, Ecole Normale Supérieure, PSL University, CNRS, 45 rue d’Ulm, Paris 75005, France}

\begin{document}

\begin{abstract}
Efficient verification of quantum computational resources is crucial as experiments advance toward fault-tolerance. Universal quantum computation can be achieved by consuming resource states through simple Pauli measurements, yet a significant gap remains between states that are easy to certify and those required for universality. We focus on \emph{Clifford-enhanced Product States}, a class of resource states obtained by applying Clifford circuits to a product of single-qubit, potentially magic, states. While essential for universal computation, the certification of such states has previously relied on query oracles that are \#P-hard to implement, leaving their efficient, oracle-free verification an open challenge.

In this work, we demonstrate that such classically intractable resource states can be efficiently verified using only Pauli measurements. 
Our protocol achieves sample- and time-efficiency in both i.i.d.\ and adversarial settings.
This work fills a gap in Pauli-based certification, providing a new practical pathway to verify resource states that drive universal Pauli-based quantum computation.
\end{abstract}

\maketitle

\section{Introduction}

Quantum computers' promise and potential to outperform classical ones for certain computational tasks
\cite{S94algorithms, G96fast, F82simulating, F86quantum, D85quantum, HCMP25vast}
has generated rapid progress in the field \cite{P18quantum, EP25mind}.
Increasingly complex devices and experiments are being realized across various platforms
\cite{NLDC22beating, AAAA25constructive, DBMA25breaking}.
However, these developments must be accompanied by the design of reliable techniques to
{certify the outputs of quantum experiments} and {verify the correctness of quantum computations}
\cite{EHWR20quantum, BPY25quantum, PYBB24benchmarking}, especially in the context of delegated quantum computing,
where quantum computation is envisioned as a cloud service provided by a potentially untrusted party
\cite{GKK19verification}.

Verification of quantum computations and certification of quantum states are closely related tasks, particularly in computational models where the computation proceeds by first preparing a specific resource state and then consuming it through a sequence of adaptive measurements. This is the paradigm of measurement-based quantum computation (MBQC). Different MBQC models are distinguished by the class of resource states they use and by the measurements required to achieve universal quantum computation.

A paradigmatic example is graph-state MBQC, where the resource state is a graph state
\cite{RB01one}. This model is universal when Pauli measurements are supplemented with single-qubit measurements in the $XY$ plane, for instance at angles such as $\pi/4$ \cite{DKP07measurement}. Thus, although graph states themselves are highly structured and well suited to certification, universal computation in this model ultimately requires non-Pauli measurements. Hypergraph-based MBQC offers a different route: by using hypergraph states as resources, universal quantum computation can be achieved using Pauli measurements only \cite{TMH18quantum}. However, the preparation of such resource states generally requires non-Clifford, multi-qubit controlled operations, which are far from the standard architectures pursued in most fault-tolerant roadmaps.

Indeed, many approaches to fault-tolerant quantum computation are instead organized around Clifford operations supplemented by non-stabilizer resources, such as $T$ gates or prepared magic states
\cite{DBMA25breaking, G23cleaner, AAAA22suppressing, MGCM25ibm, GSJ24magic}.
This motivates a third perspective: achieving universal quantum computation using only Pauli measurements, provided the resource state is generated by applying Clifford circuits to input states that may include single-qubit magic states. This is the idea behind Pauli-based MBQC \cite{SDKO07direct}. In this setting, the measurements remain experimentally simple, while the non-Clifford computational power is supplied by the initial magic resources.

This perspective shifts the verification problem toward state certification. Indeed, if a computation is driven by consuming a resource state, then certifying the produced state is a natural prerequisite for trusting the computation \cite{FFS26}. A natural framework for this task is the receive-and-measure setting, where a verifier receives quantum systems prepared by a prover and performs local measurements in order to certify their closeness to a desired target state. In this work, we focus on the particularly minimal case where the verifier is restricted to single-qubit Pauli measurements. This restriction still allows one to access arbitrary Pauli observables by measuring the relevant qubits individually and classically combining the outcomes. Such measurements are native to most qubit-based architectures and impose minimal experimental requirements.

For graph states and hypergraph states, efficient Pauli-based certification protocols are known
\cite{TM17verification, TMMM18resource, ZH19efficient, ZH18efficient, LZH23robust}.
Their efficiency relies on the stabilizer, or more generally quasi-stabilizer, structure of these states: certain Pauli measurements have deterministic outcomes on the target state, allowing the verifier to construct simple tests whose acceptance probability directly witnesses fidelity. However, the same reasoning does not apply to the resource states naturally appearing in Pauli-based MBQC with magic inputs. The presence of non-stabilizer single-qubit states destroys the deterministic Pauli correlations that make stabilizer-based certification protocols effective.
 
In this work, we focus on the general class of states obtained by applying an arbitrary Clifford circuit to an arbitrary product of single-qubit states:
\begin{equation}
    \label{eq:intro:CliffordProductStates}
    \ket{\Psi} = C\!\left(\bigotimes_{i=1}^n \ket{\psi_i}\right),
\end{equation}
where $C$ is an $n$-qubit Clifford circuit and the states $\ket{\psi_i}$ are arbitrary single-qubit states, potentially including magic states.
We refer to this family as \emph{Clifford-enhanced Product States} (\CEPS, inspired by \cite{LHD24quantum}), depicted in Fig.~\ref{fig:CPS}. This class naturally captures the resource states arising from Clifford processing of product inputs, and includes the kind of magic-state resources needed for Pauli-based universal quantum computation.
In particular, suitable choices of the input states, including magic states at appropriate locations, yield universal resource states for Pauli-based MBQC \cite{SDKO07direct} while remaining within the \CEPS~family.

\begin{figure}[h]
    \centering
    \includegraphics[width=0.8\linewidth]{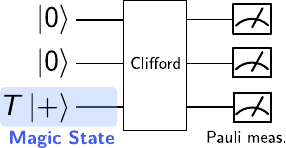}
    \caption{
    \CEPS~and their use for computation. Adaptive Pauli measurements enable universal quantum computing.}
    \label{fig:CPS}
\end{figure}

Despite their relevance, existing certification methods do not provide an efficient Pauli-based protocol for arbitrary \CEPS. First, \CEPS~are not, in general, stabilizer states, nor do they exhibit the stabilizer-like structure that enables efficient protocols for graph and hypergraph states \cite{TM17verification, TMMM18resource, ZH19efficient, ZH18efficient, LZH23robust}. Second, more general certification protocols based on single-qubit measurements have been proposed \cite{GHO25few}, with even a Pauli-only protocol in \cite{HPS24certifying} for a broad class of states,  but they assume oracle access to quantities such as Pauli amplitudes of the target state. For \CEPS, this assumption is not computationally realizable: computing amplitudes of quantum circuit output states is $\#P$-hard in the worst case \cite{A04quantuma, BFNV18quantum}, and in particular for \CEPS~\cite{YJS18quantum}. Consequently, while such protocols may be sample-efficient, they are not time-efficient for \CEPS~unless $P=P^{\#P}$. Third, \CEPS~are related to the broader framework of $t$-doped stabilizer states~\cite{LOH24Learning}, but known efficient certification methods for such states are limited to regimes where the amount of non-stabilizerness scales logarithmically with the system size~\cite{LOLH24Learning}. This does not cover arbitrary \CEPS, nor the universal Pauli-MBQC resource states contained in \CEPS, where magic states may appear at extensively many input locations.
These relationships between \CEPS~and other known classes of states are summarized in Fig.~\ref{fig:diagram}.

\begin{figure}[t]
    \centering
    \includegraphics[width=\linewidth]{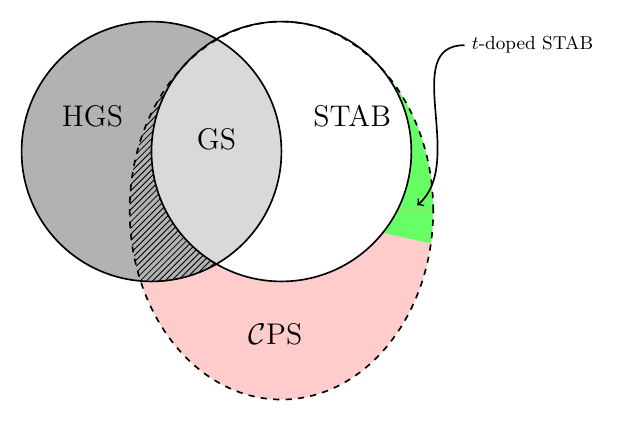}
    \caption{Relationship between Clifford-enhanced Product States (\CEPS) and other known classes: Stabilizer States (\STAB, obtained if no magic state is used), $t$-doped STAB (only $t\sim\log(n)$ regime known to be certifiable), Graph States (\GS, obtained if STAB and $C$ is made only of $CZ$), and Hypergraph States (\HGS, no intersection with \CEPS~outside \GS).}
    \label{fig:diagram}
\end{figure}

We therefore arrive at the central problem addressed in this paper: although \CEPS~form a natural family of resource states for Pauli-based quantum computation, no existing method efficiently certifies arbitrary \CEPS~using only Pauli measurements and the circuit description of the target state. Here we close this gap by introducing a certification protocol that is both sample-efficient and time-efficient: the number of required copies scales polynomially with the relevant parameters, and the associated classical pre- and post-processing can be performed in polynomial time.

To achieve this, we combine Direct Fidelity Estimation \cite{FL11direct, FOS24optimal} with a recent framework of robust fidelity witnesses \cite{CGKM21efficient}, while exploiting the Clifford structure of \CEPS~through Pauli measurement back-propagation \cite{UC24efficient}. 
The resulting protocol avoids the oracle assumptions required by more general certification methods, while still applying to resource states that may be both highly entangled and classically intractable in the computational basis.

The rest of the paper is organized as follows. In Section~\ref{section:Technical Overview}, we give an overview of the main technical ideas behind the protocol and explain how direct fidelity estimation, robust fidelity witnesses and Clifford back-propagation are combined. In Section~\ref{section:certification}, we present the certification protocol and prove its guarantees, first in the i.i.d.\ setting and then in the adversarial setting. Finally, in Section~\ref{section:VQC}, we discuss the implications of our certification result for verifying Pauli-based universal quantum computation.

\section{Technical overview}
\label{section:Technical Overview}
In this section, we present a high-level overview of the techniques used in our approach to certify \CEPS~states. It assumes prior knowledge of quantum information fundamentals recalled in Appendix \ref{sec:prelim}.
The protocol can be seen as consisting in two phases. First, a data acquisition phase during which for each received copy, a Pauli observable is sampled according to a specific {probability distribution} (see below), then conjugated by $C$ (or {back-propagated through} the circuit), which virtually applies $C^\dagger$—the inverse of the Clifford circuit of the target \CEPS~state in Eq.~(\ref{eq:intro:CliffordProductStates})—to the state, and finally measured.
Then, that data is aggregated in a specific way to form a {witness} of the fidelity, a quantity that lower-bounds the fidelity. The witness is computed with respect to the underlying product state, since the Clifford has been uncomputed on the Pauli observables via back-propagation.

The initial building block is {Direct Fidelity Estimation} (DFE). It is a method to estimate the fidelity of a state with respect to a target state using single-copy Pauli measurements sampled according to a probability distribution that is computable from the classical description of the target state \cite{FL11direct}. We present the technical preliminaries required for that in Section \ref{subsection:prelims:DFE}, taken from \cite{FOS24optimal}. In particular, we present the probability distributions on Pauli operators inspired from Importance Sampling, that we will use to derive a new randomized procedure in Section \ref{section:certification}.

Unfortunately, in DFE, the worst-case sample complexity scales exponentially in the system size. In particular, in Appendix \ref{appendix:l-1 norm of CPS}, we show that this is the case for \CEPS. Thus, the impracticality of DFE in the context of \CEPS~motivates the use of alternative approaches for efficient state certification. One such approach relies on the introduction of {robust fidelity witnesses} \cite{CGKM21efficient,UC24efficient,HPS24certifying,GHO25few}, quantities that are easier to estimate while still providing a guaranteed bound on the fidelity. We define those in \ref{subsection:prelims:witness}. These methods were originally introduced in the context of Continuous Variables quantum information: by using it in the current context, we show that they can indeed fit perfectly for Discrete Variables systems.

Finally, in \ref{subsection:prelims:backprop} we recall the {back-propagation} of Pauli observables through Clifford circuits. While this is mathematically obvious given the definition of the Clifford and Pauli groups, this is particularly relevant in the current context to virtually uncompute $C$ from $C(\bigotimes\ket{\psi_i})$. Indeed, it consists of applying $C^\dagger$ on the Pauli observables rather than on the state itself, since the verifier has the ability to perform Pauli measurements only (and not to apply a Clifford circuit). This produces the same measurement outcomes statistics and is thus perfectly equivalent in our case since we only care about the outcomes statistics and not the post-measurement state.
We use the fact that Pauli observables remain Pauli observables when conjugated by a Clifford circuit, by definition.  
This is possible because of the specific structure of \CEPS, allowing us to back-propagate {expectation values} 
rather than POVM elements as done in \cite{UC24efficient}.

\subsection{Direct Fidelity Estimation (DFE)}
\label{subsection:prelims:DFE}

We here remind the optimal method for DFE, following \cite{FOS24optimal}. 
In what follows, let $d=2^n$ be the dimension of the Hilbert space and $\mathcal{P}_n$ denote the group of $n$-qubit Pauli operators. 
We define the {characteristic function} of a pure state $\ket{\psi}$ as
\begin{equation}
    \chi_\psi(P) = \Tr[\ketbra{\psi} P].
\end{equation}
The state can be expanded in the Pauli basis as
\begin{equation}
\label{expansion}
    \ketbra{\psi} = \frac{1}{d} \sum_{P \in \mathcal{P}_n} \chi_\psi(P) P.
\end{equation}
The $\ell_1$-norm of the characteristic function is defined as
\begin{equation}
    \|\chi_\psi\|_1 = \sum_{P \in \mathcal{P}_n} |\chi_\psi(P)|.
\end{equation}
In this convention, stabilizer states satisfy $\|\chi_\psi\|_1 = d$. For $d=2$, the magic $T$-state satisfies $\|\chi_\psi\|_1 = 1 + \sqrt{2}$. More generally, for a single qubit with Bloch vector $\vec r$ we have $\|\chi_\psi\|_1 = 1+\|\vec r\|_1 \le 1+\sqrt3$, with equality for $\vec r = (1,1,1)/\sqrt3$; the $T$-state is therefore not the extremal case. Note that we use an unnormalized convention for the $\ell_1$-norm, including the identity Pauli. This differs from the normalized convention of Ref.~\cite{FOS24optimal}, where the identity contribution is removed; the two conventions are related by an elementary rescaling.

We define the probability distribution $\mathcal{D}_\psi$ over all $P \in \Pauli_n$ as
\begin{equation}
\label{DpsiP}
    \mathcal{D}_\psi(P) = \frac{|\chi_\psi(P)|}{\norm{\chi_\psi}_1}.
\end{equation}
The distribution $\mathcal{D}_\psi$ can be sampled from in $O(d)$ time.
Using these definitions, the fidelity $F(\rho, \ket{\psi}) = \Tr[\rho \ketbra{\psi}]$ is expressed as a single expectation value:
\begin{align}    
    F(\rho, \ket{\psi}) 
        &= \frac{1}{d} \sum_{P \in \Pauli_n} \chi_\psi(P)\, \Tr[\rho P] \nonumber \\
        &= \frac{1}{d} \sum_{P \in \Pauli_n} |\chi_\psi(P)|\sgn(\chi_\psi(P))\, \Tr[\rho P] \nonumber \\
        &= \frac{\norm{\chi_\psi}_1}{d} \sum_{P \in \Pauli_n} \frac{|\chi_\psi(P)|}{\norm{\chi_\psi}_1} \sgn(\chi_\psi(P)) \Tr[\rho P] \nonumber \\
        &= \mathbb{E}_{P \sim \mathcal{D}_\psi}
           \left[
           \frac{\norm{\chi_\psi}_1}{d} \, 
             \sgn(\chi_\psi(P)) \, \Tr[\rho P]
           \right].
\end{align}
Hence, the estimation of the fidelity $F(\rho, \ket{\psi})$ reduces to estimating the expectation value of a random Pauli observable drawn from $\mathcal{D}_\Psi$.
Define
\begin{equation}
    g(P) = (\norm{\chi_\psi}_1/d) \, \sgn(\chi_\psi(P)) \, \Tr[\rho P],
\end{equation}
so that $F(\rho, \ket{\psi}) = \mathbb{E}_{P \sim \mathcal{D}_\psi}[g(P)]$.

By performing repeated Pauli measurements according to $\mathcal{D}_\psi$ and averaging the outcomes, one obtains an unbiased estimator of the fidelity. To achieve additive precision $\epsilon$ and failure probability $\delta$, $N> N' \in O(d\, \epsilon^{-2} \log(1/\delta))$ copies of $\rho$ are sufficient, following Hoeffding's inequality (Lemma \ref{lemma:prelim:hoeffding}).

Although DFE achieves the information-theoretically optimal scaling for direct fidelity estimation, the sample cost grows exponentially with the system size $n$, since $d = 2^n$. Consequently, this method remains practical only for small-scale quantum systems.

\subsection{Fidelity witnesses}
\label{subsection:prelims:witness}
We hereby present the concept of fidelity witness as in \cite{CGKM21efficient, UC24efficient}. A {robust fidelity witness} $W(\rho, \ket{\psi})$ is a quantity that satisfies
\begin{equation}
\label{eq: robust fidelity witness}
    f(F(\rho, \ket{\psi}))
    \leq W(\rho, \ket{\psi}) 
    \leq F(\rho, \ket{\psi}) ,
\end{equation}
where $f$ is a continuous function such that $f(1) = 1$. 
The right-hand inequality shows that the witness is a {lower bound} on the true fidelity $F(\rho, \ket{\psi})$, while the left-hand inequality characterizes the region where the witness provides meaningful information, namely where $W(\rho, \ket{\psi}) \ge 0$.

For a target $n$-qubit product state $\ket{\psi}=\bigotimes_{i=1}^n\ket{\psi_i}$, a fidelity witness was introduced in \cite{CGKM21efficient} as
\begin{equation}
    \label{eq:prelims:witness}
    W(\rho, \ket{\psi}) 
    \coloneqq 1 - 
    \sum_{i=1}^{n}
    \big(1 - F(\rho_i, \ket{\psi_i})\big),
\end{equation}
where $\rho_i$ denotes the single-qubit reduced state of $\rho$ obtained by tracing out all subsystems but the $i$-th one.
This expression combines local fidelities into a global quantity that certifies closeness to the target product state while remaining efficiently computable from local measurements. The {robustness} property of this witness is given by \cite{CGKM21efficient}
\begin{equation}
    \label{eq:prelims:witness:robustness}
    1-n(1-F(\rho, \ket\psi)) \leq W(\rho, \ket\psi)\leq F(\rho, \ket\psi)\; .
\end{equation}
In other words, if $F(\rho,\ket\psi)=1-\alpha$, then the bound only guarantees $W(\rho,|\psi\rangle)\geq 1-n\alpha$: the infidelity may be amplified by a factor $n$ at the level of the witness. We refer to this amplification factor as the robustness loss of the witness, which is here linear in the system size. Conversely, obtaining $W(\rho,|\psi\rangle)\geq 1-\epsilon$ requires $F(\rho,|\psi\rangle)\geq 1-\epsilon/n$.

\subsection{Measurement back-propagation}
\label{subsection:prelims:backprop}
 
We use the term {back-propagation} for the classical post-processing of Pauli measurement outcomes that compensates for a preceding Clifford operation before the measurement. Formally, let $\rho$ be an $n$-qubit density operator, $C\in\Clifford_n$, and $P\in\Pauli_n$. Measuring $P$ on $C^\dagger \rho C$ yields the expectation value
\begin{equation}
\Tr\left[
(C^\dagger\rho C)P
\right]
= \Tr\left[\rho\,(C P C^\dagger)\right],
\end{equation}
by cyclicity of the trace. Thus, measuring $P$ on the uncomputed state $C^\dagger\rho C$ is equivalent in expectation value to measuring the \emph{back-propagated} observable $C P C^\dagger$ directly on $\rho$. Since $C$ is a Clifford operation, $C P C^\dagger\in\Pauli_n$, so the back-propagated measurement remains a Pauli measurement. Since we only care about the measurement outcomes statistics and not about the post-measurement states, these two approaches are totally equivalent.
In other words, “uncomputing” $C$ (i.e., applying $C^\dagger$ on $\rho$ before measuring $P$) can be replaced by measuring $C P C^\dagger$ on $\rho$ at the level of expectation values. This approach is depicted in Fig.~\ref{fig:backprop}. Back-propagation of POVM elements was already introduced in \cite{UC24efficient} in the context of verification in continuous variables quantum information processing tasks. In discrete variables, it is also crucially used in \emph{prepare-and-send} verification protocols \cite{FK17unconditionally, KKLM22unifyinga}. Its usage here in this receive-and-measure protocol showcases again the usefulness of the efficient mapping of the Pauli group through Clifford evolutions in the context of verification tasks.
\begin{figure}[h]
    \centering
    \includegraphics[width=1\linewidth]{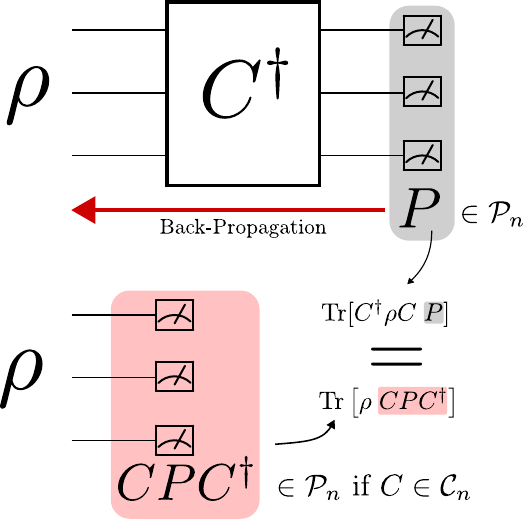}
    \caption{Measurement back-propagation: measuring observable $P$ on $C^\dagger \rho C$ yields the same outcome distribution as
     measuring $C P C^\dagger$ on $\rho$. By definition of the Clifford group, $CPC^\dagger$ is also a Pauli observable if $C\in\Clifford_n$.}
    \label{fig:backprop}
\end{figure}

\section{Quantum state certification with few Pauli measurements}
\label{section:certification}
In this section, we give a protocol to check whether a prover correctly prepares an $n$-qubit target state in the \CEPS\ class, meaning a quantum state of the form
        \(
        \ket{\Psi} = C\left(\bigotimes_{i=1}^n \ket{\psi_i}\right),
        \)
        where $C$ is a known Clifford circuit and the $\ket{\psi_i}$ are arbitrary known single-qubit states, possibly including magic states.
        More formally, 
the verifier holds the classical description of $\ket\Psi$ ($C$ and the single-qubit states) and defines a target precision $\epsilon$ and the maximum allowed failure probability $\delta$, meaning that the protocol should succeed with probability $\geq 1-\delta$. Then, we first consider the case where the verifier receives from the prover $N+1$ copies of an unknown quantum state $\rho$ on $n$ qubits, one of which is set aside unmeasured to serve as the protocol's output state (the prover is said to be \textit{i.i.d.}).
The goal of the protocol is for the verifier to use (non-adaptive) Pauli measurements and classical post-processing to output an acceptance flag if $F(\rho, \ket\Psi) \geq 1-c$, a rejection flag if $F(\rho, \ket\Psi)\leq 1-s$ where $c$ and $s$ are completeness and soundness variables parametrized by $\epsilon$ (Theorem~\ref{theorem:certification} achieves $c=\frac{\epsilon}{3n}$ and $s=\epsilon$).
The efficiency analysis of the protocol consists of computing how many samples the verifier has to ask from the prover in order to achieve this certification task with precision $\epsilon$ and success probability $\geq1-\delta$, knowing the state has size $n$. In other terms, it amounts to computing how $N$ scales with $\epsilon, \delta, n$.

If the prover is not i.i.d., the state received by the verifier might be entangled across the $N$ samples, so they cannot be considered as independent copies anymore. However the verifier has the same aim: perform Pauli measurements on some of the received samples, and leave one remaining subsystem that is accepted if its fidelity is close enough to the target state.

We introduce some notations in Section\ \ref{subsection:notations}, describe the i.i.d.\ protocol in Section\ \ref{subsection:iid} and generalize that approach to the non-i.i.d.\ setting in Section\ \ref{rmk:malicious}. In both cases, we provide theorems showing that the protocols are efficient. In the i.i.d.\ setting, $N \sim O(n^2/\epsilon^2 \; \log(1/\delta))$, while in the non-i.i.d.\ setting, we have $N\sim \tilde O(n^5/(\delta^2\epsilon^6))$, with $\tilde O$ hiding polylogarithmic factors.

\subsection{Notations}
\label{subsection:notations}
Here, we introduce relevant notations for the certification protocol.
            For
            $N\in\mathbb
            N$, let $\rho^{\otimes N}$ be $N$ copies of an unknown, possibly mixed, $n$-qubit quantum state.
            Let $\np=\sum_\qubit \norm{\chi_{\psi_\qubit}}_1\le n(1+\sqrt3)$. We define the probability distribution on qubit indices $i\in\{1,\dots,n\}$:
            \begin{equation}
            \label{mu_i}
                \mu(\qubit) \coloneqq \frac{\norm{\chi_{\psi_\qubit}}_1}\np.
            \end{equation}
            For all $i\in\{1,\dots,n\}$, we also define the probability distributions over single-qubit Pauli operators
            \begin{equation}
            \label{Dpsi_iP}
                \mathcal{D}_{\psi_i}(P) \coloneqq \frac{|\chi_{\psi_i}(P)|}{\norm{\chi_{\psi_i}}_1}.
            \end{equation}
             Finally, we define a joint probability distribution on qubit indices $i$ and single-qubit Pauli operators $P$ as
             \begin{equation}
            \label{DpsiPi}
                \pdist(\qubit, P) \coloneqq \mu(\qubit)\pdist_{\psi_\qubit}(P)=\frac{|\chi_{\psi_i}(P)|}\np.
            \end{equation}
            By construction, this joint distribution can be sampled from classically by sampling the subsystem index $i$ from $\mu(i)$, and then the single-qubit Pauli operator $P$ from $\mathcal{D}_{\psi_i}(P)$.

\subsection{The protocol in the i.i.d.\ setting}
\label{subsection:iid}

            We now state our main \CEPS\ certification protocol in the i.i.d.\ setting, which is also described in Fig.~\ref{fig:protocol:certification}.

        \begin{protocol}[Certification of \CEPS]
            \label{protocol:certification}
            Let $\ket\Psi = C\left(\bigotimes_{\qubit=1}^n\ket{\psi_\qubit}\right)$ be the target state to certify and let $\epsilon>0$ be a precision parameter. With the notations above, the protocol goes as follows:
            \begin{enumerate}
                \item The verifier receives $N+1$ copies of $\rho$ and sets the last copy aside, unmeasured; we denote it $\rho_{\mathrm{fin}}$.
                \item For each copy $\copyindex\leq N$, sample classically $(\qubit_j, P_j)\sim \pdist$, then measure $C P_j^{(i_j)} C^\dagger$ on $\rho$ (where $P^{(i)}$ denotes the $n$-qubit Pauli operator acting as $P$ on the $i$-th qubit and as identity on all the other qubits), obtaining the outcome $x_j\in\{-1, 1\}$.
                \item Compute $\bar X \coloneqq \frac{1}{N}\sum_{\copyindex=1}^N \tfrac 1 2 \mathrm{sgn}(\chi_{\psi_{\qubit_j}}(P_j)) \,x_j$.
                \item Output $\bar W\coloneqq 1-n+\np\times \bar X$ together with the unmeasured copy $\rho_{\mathrm{fin}}$. Accept if $\bar W\ge1-\frac23\epsilon$, reject otherwise.
            \end{enumerate}
        \end{protocol}

        In the i.i.d.\ setting, the held-out copy $\rho_{\mathrm{fin}}$ is simply in state $\rho$; it becomes essential in the non-i.i.d.\ setting (Section~\ref{rmk:malicious}) since it might be correlated with other samples. The held-out copy is also particularly important for the verification of computations (Section~\ref{section:VQC}) since it will be used to perform the measurements that actually drive the computation, after certification. At a high level, Protocol~\ref{protocol:certification} uses back-propagation (see Section~\ref{subsection:prelims:backprop}) to estimate single-qubit fidelities between reduced states of $C^\dag\rho C$ and the corresponding single-qubit target state $\ket{\psi_i}$ via DFE (see Section~\ref{subsection:prelims:DFE}). This DFE requires more samples when $\ket{\psi_i}$ is a magic state. To compensate for this, the qubits $i$ are randomly selected using importance sampling via the $\mu(i)$ distribution. These single-qubit fidelities are then combined to estimate a robust witness of the fidelity $F(C^\dag\rho C,C^\dag\ket\Psi)=F(\rho,\ket\Psi)$ (see Section~\ref{subsection:prelims:witness}). Throughout the protocol, the prefactor $\tfrac12$ arises because the relevant probability distributions are over single-qubit Paulis, i.e.\ the $d=2$ case of Section~\ref{subsection:prelims:DFE}.
        
        Protocol~\ref{protocol:certification} only requires single-qubit Pauli measurements and simple classical post-processing (sampling from univariate distributions), making it time efficient in $N$ and $n$. We now show that this protocol provides a certification of the target \CEPS\ $\ket\Psi$ that is also sample-efficient:

        \begin{theorem}[Efficient certification of \CEPS]
        \label{theorem:certification}
            Let $\epsilon,\delta>0$. With the notations of Protocol~\ref{protocol:certification}, assuming $N\ge \Niid$ where one can take
            \begin{equation}
            \label{eq:certification:sample iid}
            \Niid = \frac{9\np^2}{2\epsilon^2} \log\left(\frac{1}{\delta}\right)
                    \in\O\left(
                    \frac{n^2}{\epsilon^2} \log\left(\frac{1}{\delta}\right)
                    \right)
                    \; ,
            \end{equation}
                 then Protocol~\ref{protocol:certification} rejects if $F(\rho,\ket\Psi)<1-\epsilon$ and accepts if $F(\rho,\ket\Psi)\ge1-\frac\epsilon{3n}$, with probability greater than $1-\delta$.
        \end{theorem}
        The constant $\frac13$ is chosen for simplicity and can be replaced by $1-\frac2k$ for any constant $k>2$ with a constant increase in sample complexity, by replacing the acceptance condition in step 4 of Protocol~\ref{protocol:certification} by $\bar W\ge1-(1-\frac1k)\epsilon$ and setting $\lambda=\frac\epsilon k$ in the proof.

        \begin{proof}
        Recall that $\jdist$ is a joint probability distribution over qubit indices and single-qubit Pauli observables. We define the estimator
        \begin{equation}
        \label{def_f}
            f(i,P,x)\coloneqq \tfrac 1 2 \sgn(\chi_{\psi_i}(P)) \times x\;,
        \end{equation}
            where $(i,P)\sim\mathcal D$ and $x\in\{+1, -1\}$ is a random variable describing the binary outcome of measuring the observable $C P^{(i)}C^\dagger$ on $\rho$ (where $P^{(i)}$ denotes the $n$-qubit Pauli operator acting as $P$ on the $i$-th qubit and as identity on all the other qubits). Then, the quantity $\bar X$ computed at step $3$ of Protocol~\ref{protocol:certification} using $N$ samples $(i_j,P_j,x_j)$ is an empirical estimate of the expected value of $f$ over the classical randomness of $(i,P)$ and the quantum randomness of $x$, which writes $X=\mathbb E_{i,P,x}[f(i,P,x)]$.
        
        Now we have:
\begin{widetext}
\begin{align}
    X &= \mathbb E_{i,P,x}[f(i,P,x)] \nonumber \\
    &= \sum_{i=1}^n \mu(i) \sum_{P \in \mathcal{P}_1} \mathcal{D}_{\psi_i}(P) \left[ \tfrac 1 2 \sgn(\chi_{\psi_i}(P)) \Tr[\rho \, C P^{(i)} C^\dagger] \right] \nonumber \\
    &= \sum_{i=1}^n \frac{\|\chi_{\psi_i}\|_1}{\np} \sum_{P \in \mathcal{P}_1} \frac{|\chi_{\psi_i}(P)|}{\|\chi_{\psi_i}\|_1} \left[ \tfrac 1 2 \sgn(\chi_{\psi_i}(P)) \Tr[\rho \, C P^{(i)} C^\dagger] \right] \nonumber \\
    &= \frac{1}{\np} \sum_{i=1}^n \left( \frac{1}{2} \sum_{P \in \mathcal{P}_1} \chi_{\psi_i}(P) \Tr[C^\dagger \rho C \, P^{(i)}] \right) \nonumber \\
    &= \frac{1}{\np} \sum_{i=1}^n \left( \frac{1}{2} \sum_{P \in \mathcal{P}_1} \chi_{\psi_i}(P) \Tr[(C^\dagger \rho C)_i \, P] \right) \nonumber \\
    &= \frac{1}{\np} \sum_{i=1}^n \Tr\left[ (C^\dagger \rho C)_i \left( \frac{1}{2} \sum_{P \in \mathcal{P}_1} \chi_{\psi_i}(P) P \right) \right] \nonumber \\
    &= \frac{1}{\np} \sum_{i=1}^n \Tr[(C^\dagger \rho C)_i \, \ketbra{\psi_i}] \nonumber \\
    &= \frac{1}{\np} \sum_{i=1}^n F((C^\dagger \rho C)_i, \ket{\psi_i}).
\end{align}
\end{widetext}
        where we used Eq.~\eqref{def_f} and the Born rule in the second line, Eqs.~\eqref{mu_i} and~\eqref{Dpsi_iP} in the third line, the cyclicity of the trace in the fourth line, where in the fifth line $(C^\dag\rho C)_i$ denotes the $i$-th reduced subsystem of $C^\dag\rho C$, and where we used Eq.~\eqref{expansion} in the seventh line and Eq.~\eqref{fidover} (from Appendix \ref{sec:prelim}) in the last line.
        
        With Eq.~\eqref{eq:prelims:witness} we thus obtain 
        \begin{equation}
            \label{eq:equality witness sum of fid}
            W(C^\dagger\rho C,C^\dagger\Psi C) = 1-n+\np X\; .
        \end{equation}

            \begin{figure*}[t]
                \centering
                \includegraphics[width=0.9\textwidth]{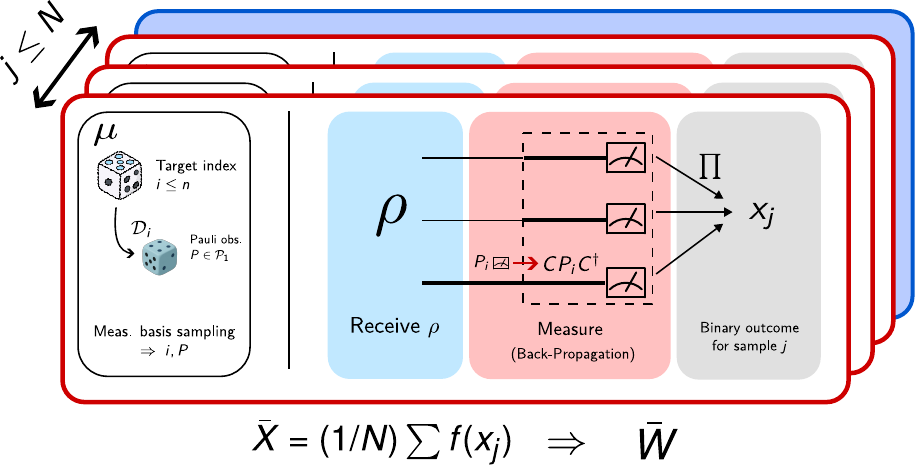} 
                \caption{Visual description of our protocol to certify $\mCEPS$. With the notations of Protocol \ref{protocol:certification}, here for copy $j$ for which we sample $\qubit, P$, we have $f(i_j,P_j,x_j) = \tfrac 1 2 \mathrm{sgn}(\chi_{\psi_{\qubit_j}}(P_j)) \,x_j$ (cf. Eq.~\eqref{def_f}). The estimate for the witness $\bar W$ is computed from the estimator $\bar X$ using $\bar W = 1-n+\np\times \bar X$.
                The samples in red boxes are used for certification, while
                the blue box is the sample left unmeasured, the output copy $\rho_{\mathrm{fin}}$.}
                \label{fig:protocol:certification}
            \end{figure*}

        Next, we show how close the empirical estimate $\bar X=\frac1N\sum_{j=1}^Nf(i_j,P_j,x_j)$ is to $X=\mathbb E_{i,P,x}[f(i,P,x)]$. Since the range of $f$ is $\max_{i,P,x}f(i,P,x)-\min_{i,P,x}f(i,P,x)\leq1$, Hoeffding's inequality (see Lemma \ref{lemma:prelim:hoeffding}) implies
        \begin{equation}
            \Pr[
            \bar X - X \geq \lambda
            ] \leq \exp
            \left(
            -2N\lambda^2
            \right),
        \end{equation}
        for any $\lambda>0$, and the same bound holds for $\Pr[X - \bar X \geq \lambda]$.
        Furthermore, by Eq.\ (\ref{eq:equality witness sum of fid}) we have $W=1-n +\np X$, so
        \begin{equation}
            \begin{aligned}
            \mathrm{Pr}[
            \bar W - W \geq \lambda
            ] &=
            \mathrm{Pr}\left[
            \bar X - X \geq \dfrac{\lambda}{\np}
            \right] \\
            &\leq
            \exp{
            \left(
            -\dfrac{2N \lambda^2}{\np^2}
            \right)
            },
            \end{aligned}
        \end{equation}
        and similarly for $\Pr[W - \bar W \geq \lambda]$, each of which is smaller than $\delta>0$ as soon as $N\ge N(\lambda)$ where
            \begin{equation}
                N(\lambda)\coloneqq\frac{\np^2}{2\lambda^2}\log\left(\dfrac{1}{\delta}\right)\; .
            \end{equation}
 
        By cyclicity of the trace we have $F(C^\dagger\rho C,C^\dagger\Psi C)=F(\rho,\ket\Psi)$ so Eq.~\eqref{eq:prelims:witness:robustness} implies
        \begin{equation}
            1-n(1-F(\rho,\ket\Psi)) \leq W(C^\dagger\rho C,C^\dagger\Psi C) \leq F(\rho,\ket\Psi).
        \end{equation}
        Hence, for $N\ge N(\lambda)$, each of the events $\bar W \leq W + \lambda$ and $\bar W \geq W - \lambda$ holds with probability greater than $1-\delta$. For soundness, $\bar W \leq W+\lambda \leq F(\rho,\ket\Psi)+\lambda$, so $F(\rho,\ket\Psi)<1-\epsilon$ implies $\bar W<1-\epsilon+\lambda$ with probability greater than $1-\delta$. For completeness, $\bar W \geq W-\lambda \geq 1-n(1-F(\rho,\ket\Psi))-\lambda$, so $F(\rho,\ket\Psi)\ge1-\frac\epsilon{3n}$ implies $\bar W\ge1-\epsilon/3-\lambda$ with probability greater than $1-\delta$. Setting $\lambda=\epsilon/3$ concludes the proof, with $\Niid=N(\epsilon/3)\in\mathcal O(\frac{n^2}{\epsilon^2}\log(\frac{1}{\delta}))$ since $m\le n(1+\sqrt3)$.
        \end{proof}
        As a result, we can efficiently certify any $n$-qubit $\mCEPS$ using only (non-adaptive) single-qubit Pauli measurements on a number of samples that scales quadratically in $n$ and with efficient sampling and post-processing of measurement outcomes.

\subsection{Extension to the non-i.i.d.\ setting}

\label{rmk:malicious}
 While Protocol \ref{protocol:certification} is formulated for the i.i.d.\ setting, it naturally extends to the non-i.i.d.\ (adversarial) scenario thanks to its non-adaptive and incoherent measurement structure. We hereby use the results and procedure presented by \cite{FKMO24learning}. By incorporating a classical preprocessing phase consisting of a random permutation of the subsystems, the verifier can achieve the same fidelity guarantees for a held-out state, even when the prover provides an arbitrarily correlated joint state.
 The resulting protocol is described in Protocol \ref{protocol:verification} and an intuitive depiction can be found in Fig~\ref{fig:adversarial}.
\begin{protocol}[Certification of \CEPS, beyond i.i.d.]
\label{protocol:verification}

Let $\ket\Psi = C\left(\bigotimes_{\qubit=1}^n\ket{\psi_\qubit}\right)$ be the target state to certify, $\epsilon>0$ be a threshold of acceptance. Also, let $N_1,N_2\in\mathbb N$ be such that $N_1>N_2$. Denoting by $\rho^{1\ldots N_1}\in(\Hilbert^{\otimes n})^{\otimes N_1}$ the state prepared by the prover,
the protocol goes as follows. The verifier first chooses a random partition of the set $\{1, \ldots, N_1\}$ into three subsets with respective size and labels:
size $N_{2}$ (\texttt{test}), $N_1-N_2-1$ (\texttt{discard}), and $1$ (\texttt{keep}). Then:
\begin{itemize}
    \item Discard the received registers $j\in \texttt{discard}$.
    \item Run steps 2--4 of Protocol \ref{protocol:certification} on the registers $j\in \texttt{test}$ (with $N=N_2$). Output the resulting acceptance/rejection flag.
    \item Keep the register $j\in\texttt{keep}$ unmeasured, denoted $\rho_{\mathrm{fin}}$, together with the above flag.
\end{itemize}
\end{protocol}

\begin{figure}
    \centering
    \includegraphics[width=\linewidth]{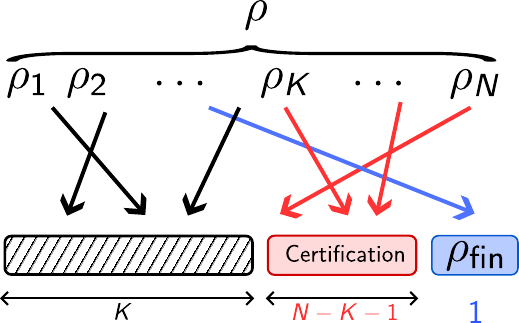}
    \caption{Illustration of the protocol to run against an adversarial prover. The verifier receives different $n-$qubit subsystems of a global system $\rho^{1\ldots N}\in(\Hilbert^{\otimes n} )^{\otimes N}$.
    The verifier chooses a random permutation of the subsystems, discards $K$ of them, applies the certification protocol on $N-K-1$ of them, and leaves the last sample as output state.
    In the text, $N$ is replaced by $N_1$ and $K$ is replaced by $N_2$.}
    \label{fig:adversarial}
\end{figure}

 To ensure security with parameters $\epsilon, \delta$, sample complexity in this setting increases but remains polynomial in the system size $n$, ensuring that the protocol is safe for deployment in untrusted environments such as delegated quantum computing.
 This is captured by Theorem \ref{theorem:verification security}.
 We provide the formal lemma and complexity analysis for this generalization in Appendix \ref{section:verification}.
     \begin{theorem}[Efficient certification of \CEPS, beyond i.i.d.]
     \label{theorem:verification security}
            Let $\epsilon,\delta>0$. With the notations of Protocol~\ref{protocol:verification}, assuming 
            $N_2\ge \Niid$, with $\Niid$ as in Theorem~\ref{theorem:certification},
            and
            \(
            N_1\in\Omega\!\left(
            \frac{n}{\delta^2\epsilon^2}\,\Niid^{\,2}\,
            \log^2\!\Big(\frac{\Niid}{\delta}\Big)
            \right),
            \)
            then Protocol~\ref{protocol:verification} rejects if $F(\rho_{\mathrm{fin}},\ket\Psi)<1-\epsilon$ and accepts if $F(\rho_{\mathrm{fin}},\ket\Psi)\ge1-\frac\epsilon{3n}$, with probability greater than $1-\delta$.
        \end{theorem}

        As a corollary of the above, verifying an $n$-qubit state of the \CEPS~class with precision $\epsilon$ and failure probability at most $\delta$ requires a total number of samples $N_1\geq \Nnoniid$ where
        \begin{equation}
            \label{eq:verification:scaling}
        \Nnoniid \in \tilde\O\left(\frac{n^5}{\delta^2\epsilon^6}\right)\;,
        \end{equation}
        with $\tilde O$ hiding polylogarithmic factors in $n,\epsilon, \delta$.

\section{Verification of universal quantum computations}
\label{section:VQC}

In this section, we show that the protocol introduced in Section~\ref{section:certification} can be used to verify universal quantum computations. This follows from the fact that certain states in the $\mCEPS$ class serve as resource states for universal computation.

We consider the Pauli-MBQC model \cite{SDKO07direct}, which consists of the standard Measurement-Based Quantum Computation (MBQC) model where nodes are prepared in magic states (such as $\ket{T} = T\ket{+}$) and processed using only adaptive Pauli measurements. This model is particularly relevant for our protocol for three reasons:
\begin{enumerate}
    \item It is universal for quantum computation \cite{SDKO07direct}.
    \item The underlying resource state $\ket{\Psi}$ belongs to the $\mCEPS$ class, as it is a tensor product of single-qubit states (including magic states) entangled by a Clifford circuit (typically a collection of $CZ$ gates).
    \item The computation relies solely on Pauli measurements, introducing no extra measurement requirements beyond those needed for certification.
\end{enumerate}
Below, we state our protocol for verifying computations in the Pauli-MBQC model. It is directly based on the certification protocols presented in the previous section: in the i.i.d.\ setting the verifier uses Protocol~\ref{protocol:certification}, while in the non-i.i.d.\ setting it uses Protocol~\ref{protocol:verification}, with the corresponding sample costs. We use the term ``systems'' rather than ``samples'' in this protocol to avoid suggesting an i.i.d.\ assumption.
\begin{protocol}[Verification of Quantum Computations \textit{via} Pauli measurements]
    \label{protocol:verification of UQC}

    Let $\ket\Psi$ be the $\mCEPS$ resource state associated with the target computation.

    \begin{enumerate}

        \item \textbf{Certification Phase:}
        The verifier runs Protocol~\ref{protocol:certification} with $N = \Niid$ in the i.i.d.\ setting (or Protocol~\ref{protocol:verification} with $(N_1,N_2) = (\Nnoniid,\Niid)$ in the non-i.i.d.\ setting), with $\Niid$ and $\Nnoniid$ as in Eqs.~\eqref{eq:certification:sample iid} and \eqref{eq:verification:scaling}, respectively. In both cases the protocol leaves one system $\rho_{\mathrm{fin}}$ unmeasured. If the certification protocol rejects, \textbf{Abort}.

        \item \textbf{Computation Phase:}
        If the certification protocol accepts, perform the target adaptive Pauli measurements on $\rho_{\mathrm{fin}}$ to obtain the result $\rho_{\mathrm{res}}$.

    \end{enumerate}
\end{protocol}

The guarantee on the closeness of the prepared state to the target state—measured in the trace distance $D$—can be lifted to a guarantee on the closeness of the distributions of the computation outcomes via the Fuchs--van de Graaf inequalities (Eq.~\eqref{eq:prelim:fuchsvandegraaf} in Appendix \ref{sec:prelim}).
Completeness is stated for the honest i.i.d.\ prover, since an honest prover is always i.i.d. Soundness holds for the setting selected in the certification phase, i.e. either i.i.d.\ or non-i.i.d.

\begin{theorem}[Efficient verification of quantum computations \textit{via} Pauli measurements]
    Let $\epsilon,\delta>0$, and let $\ket{\Psi}$ be the $\mCEPS$ resource state associated with a target Pauli-MBQC computation with ideal output $\ket{\Psi_{\rm out}}$. Protocol~\ref{protocol:verification of UQC} satisfies the following two properties:
    \begin{enumerate}
        \item \textbf{Completeness:} for an honest i.i.d.\ prover preparing systems in a state $\rho$ such that $F(\rho,\ket{\Psi})\geq 1-\epsilon/(3n)$, the protocol accepts with probability at least $1-\delta$.

        \item \textbf{Soundness:} for any prover, except with probability at most $\delta$, if the protocol accepts, then the output guarantee~\eqref{eq:verification-output-guarantee} holds.
    \end{enumerate}
    In the acceptance case, the output guarantee is
    \begin{equation}
        \label{eq:verification-output-guarantee}
        D(\rho_{\mathrm{res}},\ketbra{\Psi_{\rm out}})\leq \sqrt{\epsilon}.
    \end{equation}
\end{theorem}

\begin{proof}
    We first prove completeness. Consider an honest i.i.d.\ prover preparing systems in a state $\rho$ such that $F(\rho,\ket{\Psi})\geq 1-\epsilon/(3n)$. By Theorem~\ref{theorem:certification}, the certification Protocol~\ref{protocol:certification} accepts with probability at least $1-\delta$. Therefore, Protocol~\ref{protocol:verification of UQC} proceeds to the computation phase with probability at least $1-\delta$.

    We now prove soundness. Upon acceptance of the certification Protocol~\ref{protocol:certification} or Protocol~\ref{protocol:verification}, we have the guarantee that
    \begin{equation}
        F(\rho_{\mathrm{fin}}, \ket\Psi) \geq 1-\epsilon
    \end{equation}
    with probability at least $1-\delta$.

    It remains to show that, in both cases, this implies the output guarantee. By the Fuchs--van de Graaf inequalities (Eq.~\eqref{eq:prelim:fuchsvandegraaf}), the trace distance between the certified resource state and the ideal one is bounded by
    \begin{equation}
        D(\rho_{\mathrm{fin}}, \ket\Psi) \leq \sqrt{1 - F(\rho_{\mathrm{fin}}, \ket\Psi)} \leq \sqrt{\epsilon}.
    \end{equation}
    The target computation can be viewed as a CPTP map $\mathcal{C}: \mathcal{H}_{2^n} \to \mathcal{H}_{2^k}$, where $k$ is the number of output qubits. This map includes the sequence of adaptive Pauli measurements and the tracing out of measured nodes\footnote{Although individual measurement outcomes are non-deterministic, the existence of a valid flow on the graph ensures that adaptive byproduct operators can compensate for the randomness of these outcomes, thereby defining a deterministic physical process from the resource state to the output space.}.

    By contractivity of the trace distance under CPTP maps, the distance between the final output state $\rho_{\mathrm{res}} = \mathcal{C}(\rho_{\mathrm{fin}})$ and the ideal output state $\ketbra{\Psi_{\rm out}} = \mathcal{C}(\ketbra{\Psi})$ is bounded by the distance between the initial resource states:
    \begin{align}
        D(\rho_{\mathrm{res}}, \ketbra{\Psi_{\rm out}})
        &= D(\mathcal{C}(\rho_{\mathrm{fin}}), \mathcal{C}(\ketbra{\Psi})) \nonumber \\
        &\leq D(\rho_{\mathrm{fin}}, \ket{\Psi}) \nonumber \\
        &\leq \sqrt{\epsilon}. \nonumber
    \end{align}
    This proves the output guarantee.
\end{proof}

\section{Conclusion}
\label{section:conclusion}

In this work, we have introduced an efficient and experimentally feasible protocol for certifying \emph{Clifford-enhanced Product States} ($\mCEPS$) using Pauli measurements only, \textit{i.e.} the minimal hardware requirement for a verifier in the receive-and-measure setting. By providing a practical verification scheme for this class, we bridge the gap between states that are easily certifiable and the resource states central to several fault-tolerant roadmaps \cite{DBMA25breaking, G23cleaner, AAAA22suppressing, MGCM25ibm, GSJ24magic}. 
While certain universal resources like hypergraph states can be verified using stabilizer-based techniques, they require multi-qubit non-Clifford gates to be prepared.
In contrast, the $\mCEPS$ class only requires single-qubit preparation and Clifford operations.  This class contains states that are useful for universal quantum computation, once supplemented by (adaptive) Pauli measurements.

The certification of the $\mCEPS$ class is inherently difficult due to two primary bottlenecks. First, the \#P-hard structure of these states makes the calculation of Pauli amplitudes classically intractable, rendering time-efficient certification via standard methods—such as those in \cite{HPS24certifying, GHO25few}—impossible. Second, the potentially high $\ell_1$-norm of their characteristic function (see Appendix \ref{appendix:l-1 norm of CPS}) implies that standard Direct Fidelity Estimation (DFE) \cite{FOS24optimal} would require exponential sample complexity. We have demonstrated that these barriers are not insurmountable. Indeed, combining DFE, robust fidelity witnesses and measurement back-propagation allowed to bypass the exponential time and/or sample complexity required by the previous methods, resulting from the entangled and non-stabilizer structure of \CEPS. Achieving so while remaining in the restricted setting of Pauli measurements only is particularly interesting because it enables clients with minimal hardware to verify universal quantum computations in the Pauli-MBQC model of \cite{SDKO07direct}.
Furthermore, by extending this approach to the adversarial scenario with tools for non-i.i.d.\ learning \cite{FKMO24learning}, we provide a protocol that remains robust in adversarial settings such as delegated quantum computing.

Finally, this work opens several promising avenues for future research. The first one starts with the structural resemblance between $\mCEPS$ and $t$-doped stabilizer states.
More specifically, \CEPS~can be seen as $t$-doped stabilizer state with only the last Clifford circuit being non-trivial. Here we show two things. These states are not a curiosity: this single non-trivial Clifford is enough to enable universal quantum computation, per the Pauli-MBQC model. We also show how to certify them for any $t$ ($\leq n$): it is, to the best of our knowledge, the first case of doped stabilizer states certifiable while $t$ is not strictly logarithmic.
Our work thus suggests the possibility of efficiently certifying doped stabilizer states beyond the logarithmic regime.

The second one suggests to re-frame our problem in the context of Hypothesis Testing and develop tools for Sampling Without Replacement, similar to ref.\ \cite{LZH22significance}, that could be adapted in our case. Indeed, it could potentially yield to significant improvement in the sample complexity (both in i.i.d.\ and non-i.i.d.\ settings), since ref.\ \cite{LZH22significance} has helped achieving an optimal scaling for graph and hypergraph state certification \cite{LZH23robust}. 

A third direction is to understand how far the present back-propagation-based approach extends beyond the magic-state setting considered here. The key structural feature used in our protocol is that the non-stabilizer resource is injected at the level of single-qubit input states while the circuit is Clifford and the measurements are Pauli. As a consequence Pauli measurements can be efficiently back-propagated through the Clifford circuit while remaining Pauli. It would be interesting to identify whether other resource-theoretic models of quantum computation share an analogous structure, for instance coherence-based models, and to determine whether their relevant resource states admit similarly efficient certification protocols with minimal measurement requirements.

\section*{Acknowledgements}

UC thanks D.\ Leichtle, L.\ Lewis and Y.\ Quek for inspiring discussions.
SAS acknowledges H.\ Thomas and V.\ Upreti for feedbacks on figures and the initial draft of the work, and O. Fawzi and R. Salzmann for suggesting the computation in Appendix \ref{appendix:l-1 norm of CPS}.
HO, MG, SAS, TM acknowledge funding from the Hybrid Quantum Initiative (HQI) supported by France 2030 under ANR grant ANR-22-PNCQ-0002.
UC acknowledges funding from the European Union’s Horizon Europe Framework Programme (EIC Pathfinder Challenge project Veriqub) under Grant Agreement No.~101114899.

\bibliographystyle{linksen}
\bibliography{bibliography}

\appendix

\section{Technical preliminaries}
\label{sec:prelim}
In this work, we consider systems composed of qubits whose state is a normalized vector in a two-dimensional Hilbert space $\Hilbert = \mathbb C ^2$, with composite systems described by tensor products of such spaces. For a more general introduction to quantum information science, we refer the reader to \cite{NC00quantum}.

\subsection{Fidelity and trace distance}
\label{subsection:prelims:Fidelity}

To quantify the closeness between two quantum states $\rho$ and $\sigma$, two fundamental quantities are widely used: \emph{fidelity} and \emph{trace distance}.

The fidelity is defined as
\[
F(\rho, \sigma) = \left( \Tr \sqrt{\sqrt{\rho}\,\sigma\,\sqrt{\rho}} \right)^2 .
\]
It measures the overlap between $\rho$ and $\sigma$, and satisfies $0 \leq F(\rho, \sigma) \leq 1$, with equality $F(\rho, \sigma) = 1$ if and only if $\rho = \sigma$. When one of the two states is pure, say $\sigma = \ketbra{\psi}$, then we write with an abuse of notation $F(\rho, \ket\psi)$ and it simplifies to
\begin{equation}
\label{fidover}
F(\rho, \ket{\psi}) = \Tr[\rho\ketbra\psi]= \bra{\psi}\rho\ket{\psi}.
\end{equation}

In the context of the Pauli basis $\mathcal{P}_n$, it is important to note that the operators are orthogonal with respect to the Hilbert-Schmidt inner product, such that $\Tr(PQ) = d \delta_{PQ}$ for $P, Q \in \mathcal{P}_n$. Consequently, any operator $A$ (and thus any density matrix) is expanded as $A = \frac{1}{d} \sum_P \Tr(AP) P$. This explains the $1/d$ factor appearing in the Pauli expansion of the fidelity:
\begin{align}
F(\rho, \ket{\psi}) &= \Tr\left[ \rho \left( \frac{1}{d} \sum_{P \in \mathcal{P}_n} \chi_\psi(P) P \right) \right] 
\\
&= \frac{1}{d} \sum_{P \in \mathcal{P}_n} \chi_\psi(P) \Tr[\rho P].
\end{align}

The trace distance is defined as 
\begin{equation}    
D(\rho, \sigma) = \frac{1}{2}\|\rho - \sigma\|_1 = \frac{1}{2}\Tr|\rho - \sigma|,
\end{equation}
and constitutes a metric on the space of operators. It also admits a direct operational interpretation: it characterizes the largest possible difference in outcome probabilities that any measurement can produce between the two states,
\begin{equation}
D(\rho, \sigma) = \max_{0 \le E \le \mathbb I}
\abs{\Tr[E(\rho - \sigma)]}
\end{equation}
Here the maximization is taken over all measurement effects $E$ such that $0 \le E \le \mathbb I$, corresponding to the "yes" outcome of a two-outcome POVM $\{E, \mathbb I - E\}$. Hence, the trace distance quantifies the maximal statistical distinguishability between $\rho$ and $\sigma$: it is the greatest bias achievable when attempting to tell the two states apart using any physically allowed measurement.

The two quantities are closely related through the \emph{Fuchs--van de Graaf inequalities} \cite{FG97cryptographic}:
\begin{equation}
\label{eq:prelim:fuchsvandegraaf}
1 - \sqrt{F(\rho, \sigma)} \leq D(\rho, \sigma) \leq \sqrt{1 - F(\rho, \sigma)}.
\end{equation}
These inequalities provide a quantitative bridge between the geometric notion of overlap and the operational notion of distinguishability. In particular, if the fidelity between $\rho$ and $\sigma$ is close to one, then the trace distance is necessarily small, ensuring that the two states are practically indistinguishable. 

\subsection{Concentration bounds}
\label{subsection:prelim:concentration}
A central task in quantum characterization and verification is the estimation of the fidelity between a prepared quantum state $\rho$ and a known pure target state $\ket{\Psi}$. Though the fidelity is rigorously defined mathematically, in practice one wishes to estimate empirically it by measuring observables and analyze the outcomes. Among the crucial tools to evaluate the accuracy of these methods, we find concentration bounds for empirical averages such as Hoeffding's inequality, presented in Lemma \ref{lemma:prelim:hoeffding}.
\begin{lemma}[Hoeffding's inequality]
\label{lemma:prelim:hoeffding}
    Let $\epsilon>0, N>0$, and
    $x_1, \ldots, x_N$ be i.i.d random variables sampled from a probability distribution $\dist$. On this distribution, let $g:\mathbb R \rightarrow \mathbb R$ be a bounded function, and define $G = \sup_{x\in\mathrm{supp}(\dist)}g(x)-\inf_{x\in\mathrm{supp}(\dist)}g(x)$. Then, the probability that the empirical average $\bar X = \frac 1 N \sum_{i=1}^N g(x_{i})$ exceeds the expected value $X=\mathbb E_{x\sim\dist}[g(x)]$ by more than $\epsilon$ decreases exponentially as
    \begin{equation}
        \Pr\left[
        \bar X
        -
        X
        \geq \epsilon
        \right]
        \leq
        \exp{
        \left(
        -2\dfrac{N\epsilon^2}{G^2}
        \right)
        }\; .
    \end{equation}
    The same bound holds for the lower tail $\Pr[X - \bar X \geq \epsilon]$.
\end{lemma}

\section{$\ell_1$-norm of CPS states}
\label{appendix:l-1 norm of CPS}

An $n$-qubit state $\ket{\Psi} \in \mCEPS$ is of the form $C\left(\bigotimes_{i=1}^n\ket{\psi_i}\right)$. We analyze the $\ell_1$-norm of the characteristic function $\chi_{\Psi}$ as defined in Section \ref{subsection:prelims:DFE}:
\begin{equation}
    \norm{\chi_\Psi}_{1} = \sum_{P \in \mathcal{P}_n} \abs{\Tr[P\rho]}
\end{equation}
where $\rho = \ketbra\Psi$. By the cyclicity of the trace and the fact that Clifford circuits $C$ permute Pauli operators (mapping $P$ to $\tilde{P} = C^\dagger P C \in \mathcal{P}_n$), we have:
\begin{align}
    \norm{\chi_\Psi}_{1} 
    &= \sum_{P \in \mathcal{P}_n} \abs{ \Tr\left[ C^\dagger P C \left( \bigotimes_{i=1}^n \ketbra{\psi_i} \right) \right] } \nonumber \\
    &= \sum_{\tilde{P} \in \mathcal{P}_n} \abs{ \Tr\left[ \left( \tilde P_1 \otimes \cdots \otimes \tilde P_n \right) \left( \bigotimes_{i=1}^n \rho_i \right) \right] } \nonumber \\
    &= \sum_{\tilde{P} \in \mathcal{P}_n} \prod_{i=1}^n \abs{ \Tr[ \tilde P_i \rho_i ] }.
\end{align}
Because the sum over $\mathcal{P}_n$ is a sum over all possible combinations of single-qubit Paulis $\{\mathbb{I}, X, Y, Z\}$, the product and the sum commute:
\begin{equation}
    \norm{\chi_\Psi}_{1} = \prod_{i=1}^n \left( \sum_{P \in \{ \mathbb{I}, X, Y, Z \}} |\Tr(P \rho_i)| \right) = \prod_{i=1}^n \norm{\chi_{\psi_i}}_1.
\end{equation}
For a single qubit with Bloch vector $\vec r$, we have $\norm{\chi_{\psi_i}}_1 = 1 + \norm{\vec r}_1 \le 1 + \sqrt{3}\,\norm{\vec r}_2 \le 1+\sqrt{3}$ by Cauchy--Schwarz, with equality for $\vec r = (1,1,1)/\sqrt{3}$. Hence, in the worst case,
\begin{equation}
    \norm{\chi_\Psi}_{1} \le (1 + \sqrt{3})^n.
\end{equation}
A natural magic-state input already exhibits such exponential growth. As a main illustrative example, consider each $\rho_i$ to be a magic $T$-state $\rho_T$. For the $T$-state, the Pauli expectations are $|\Tr(\mathbb{I}\rho_T)|=1$, $|\Tr(X\rho_T)|=1/\sqrt{2}$, $|\Tr(Y\rho_T)|=1/\sqrt{2}$, and $|\Tr(Z\rho_T)|=0$. Thus, for a single qubit:
\begin{equation}
    \norm{\chi_{\rho_T}}_1 = 1 + \frac{1}{\sqrt{2}} + \frac{1}{\sqrt{2}} + 0 = 1 + \sqrt{2}.
\end{equation}
Substituting this back into the product for $n$ qubits, we obtain:
\begin{equation}
    \norm{\chi_\Psi}_{1} = (1 + \sqrt{2})^n.
\end{equation}
Using the fact that $d = 2^n$, we can rewrite this $T$-state value as $\norm{\chi_\Psi}_{1} = d \left( \frac{1+\sqrt{2}}{2} \right)^n$, while the worst case reads $\norm{\chi_\Psi}_{1} \le d \left( \frac{1+\sqrt{3}}{2} \right)^n$. This confirms that the $\ell_1$-norm grows exponentially with $n$, and for stabilizer states (where $\norm{\chi_{\psi_i}}_1 = 2$), it correctly recovers $\norm{\chi_\Psi}_1 = 2^n = d$.

\section{State verification in the adversarial scenario}
\label{section:verification}

In many quantum verification tasks, security or correctness analyses rely on the assumption that all copies of the quantum state are produced independently and identically distributed (i.i.d.).
In an adversarial scenario, however, this assumption may fail: the subsystems across different rounds can be arbitrarily correlated.
The following result, adapted directly from~\cite[Theorem~3]{FKMO24learning}, shows that for protocols requiring non-adaptive, single-copy (incoherent) measurements only, the same estimation guarantees obtained in the i.i.d.\ case can still be achieved in the general, non-i.i.d.\ case — at the price of requiring a larger number of total samples, which is captured by Lemma \ref{lemma:prelim:iid}.
\begin{lemma}[State Certification in the Non-i.i.d. Setting \cite{FKMO24learning}]
\label{lemma:prelim:iid}
Let $\mathcal{A}$ be a certification protocol for an $n$-qubit target state $\ket{\Psi}$ (where $d=2^n$) that consumes $\Niid$ independent copies and satisfies the following conditions:

\begin{itemize}
    \item 
    $\Pr\left[\mathcal{A}(\rho^{\otimes \Niid}) = \Accept \; \middle|\;  F(\rho, \ket{\Psi}) < 1-\epsilon\right] \leq \delta$
    \item 
    $\Pr\left[\mathcal{A}(\rho^{\otimes \Niid}) = \Reject \; \middle|\;  F(\rho, \ket{\Psi}) \geq 1-\frac{\epsilon}{3n}\right] \leq \delta$ 
\end{itemize}

Now consider a protocol $\mathcal{B}$ acting on an arbitrary (potentially entangled) joint state $\rho^{1\ldots \Nnoniid} \in (\mathcal{H}^{\otimes n})^{\otimes \Nnoniid}$ . Protocol $\mathcal{B}$ proceeds as follows :
\begin{enumerate}
    \item Apply a random permutation to the $\Nnoniid$ subsystems .
    \item Partition the registers into three sets: a test set of size $\Niid$, a single held-out subsystem $\rho_{\mathrm{fin}}$, and a discarded set of size $\Nnoniid - \Niid - 1$ .
    \item Execute $\mathcal{A}$ on the test set and output the resulting $\Accept/\Reject$ flag .
\end{enumerate}

Then, protocol $\mathcal{B}$ inherits the performance guarantees of $\mathcal{A}$ relative to the held-out state $\rho_{\mathrm{fin}}$ :
\begin{align}
    \Pr\left[\mathcal{B}(\rho) = \Accept \; \middle|\;  F(\rho_{\mathrm{fin}}, \ket{\Psi}) < 1-\epsilon\right] &\leq \delta \text{,} \\
    \Pr\left[\mathcal{B}(\rho) = \Reject \; \middle|\;  F(\rho_{\mathrm{fin}}, \ket{\Psi}) \geq 1-\frac{\epsilon}{3n}\right] &\leq \delta \text{,} 
\end{align}
provided the total sample complexity satisfies :
\begin{equation}
\Nnoniid \in \Omega\left( \frac{n}{\delta^2\epsilon^2} \Niid^2 \log^2\left(\frac{\Niid}{\delta}\right) \right) \text{.} 
\end{equation}
\end{lemma}

We have proven the security of Protocol \ref{protocol:certification} under the assumption that the incoming state is of the form $\rho^{\otimes N}$, meaning that the states are produced in an independent and identically distributed way. In an adversarial scenario, this is not necessarily the case. In Protocol \ref{protocol:verification}, we thus increase the number of samples and randomly partition the subsystems before applying Protocol \ref{protocol:certification} on a subset of registers in order to use the results of lemma \ref{lemma:prelim:iid}. The proof of Theorem \ref{theorem:verification security} is thus stated hereby.

        \begin{proof}
            The theorem is easily proved using lemma \ref{lemma:prelim:iid}, as it is a particular instance of state certification in the non-i.i.d.\ setting using non-adaptive and incoherent measurements according to the definitions of \cite{FKMO24learning}. Indeed, the subroutine Protocol \ref{protocol:certification} involves Pauli observables that can be sampled in advance (according to the mentioned probability distribution), therefore non-adaptive. The observables are measured on the individual subsystems separately, therefore the measurements are incoherent. Lemma \ref{lemma:prelim:iid} thus applies and sets the scaling of $\Nnoniid$ for the verification procedure
            in function of $\Niid$ of the certification procedure in order to reach the same success probability $1-\delta$ for the same target precision $\epsilon$. By simply plugging the scaling of $\Niid$ in the formula of $\Nnoniid$, we get Eq.~\eqref{eq:verification:scaling}.
        \end{proof}

\end{document}